\documentclass{article}
 
\usepackage{amsmath,amssymb,amsthm}
\usepackage{microtype,xstring}
\usepackage{booktabs,tabularx,enumerate}
\usepackage{url,xspace,soul,color,hyperref}
\usepackage{graphicx,paralist,subcaption}
\usepackage{enumerate,paralist,xspace,setspace}
\usepackage[a4paper]{geometry}

\newtheorem{theorem}{Theorem}
\newtheorem{lemma}[theorem]{Lemma}
\newtheorem{corollary}[theorem]{Corollary}
\newtheorem{property}[theorem]{Property}
\newtheorem{claim}[theorem]{Claim}
\title{Coloring outerplanar graphs and planar 3-trees\\with small monochromatic components\footnote{This work started at the Graph and Network Visualization (GNV) workshop of IISA '16.}}

\author{Michael~A.~Bekos$^1$ \and Carla~Binucci$^2$ \and Michael~Kaufmann$^1$ \and Chrysanthi~Raftopoulou$^3$ \and Antonios~Symvonis \and Alessandra~Tappini$^2$
\\
\medskip
\\
$^1$Wilhelm-Schickhard-Institut f\"ur Informatik, Universit\"at T\"ubingen, Germany\\
\texttt{\{bekos,mk\}@informatik.uni-tuebingen.de}
\smallskip\\
$^2$Dipartimento di Ingegneria, Universit\`a degli Studi di Perugia\\
\texttt{carla.binucci@unipg.it, alessandra.tappini@studenti.unipg.it}
\smallskip\\
$^3$School of Applied Mathematical \& Physical Sciences, NTUA, Greece\\
\texttt{symvonis@math.ntua.gr, crisraft@mail.ntua.gr}
}
\date{} 

\begin{document}

\maketitle

\begin{abstract}
In this work, we continue the study of vertex colorings of graphs, in which adjacent vertices are allowed to be of the same color as long as each \emph{monochromatic connected component} is of relatively small cardinality. We focus on colorings with two and three available colors and present improved bounds on the size of the monochromatic connected components for two meaningful subclasses of planar graphs, namely maximal outerplanar graphs and complete planar 3-trees.
\end{abstract}

%=================================================================
\section{Introduction}
\label{sec:introduction}
%=================================================================
Colorings of graphs form one of the most fascinating and well-studied topics in graph theory. Early results date back to 1852~\cite{DBLP:journals/tis/Restivo04} and since then several variants have been proposed and studied over the~years (e.g., vertex colorings~\cite{DBLP:journals/jal/BeigelE05,DBLP:journals/siamcomp/BjorklundHK09,DBLP:conf/cocoon/FominGS07,DBLP:journals/ipl/Lawler76,DBLP:journals/ipl/Wilf84}, edge colorings~\cite{DBLP:journals/jal/ChrobakY89,DBLP:journals/algorithmica/ColeK08,Vizing64}, total colorings~\cite{DBLP:journals/jct/McDiarmidR93,DBLP:journals/dmgt/WangW04,DBLP:journals/jco/ZhangHL15}); for a survey refer to~\cite{Pardalos1999}. In the classical vertex coloring problem, the vertices of a given graph must be colored with a minimum number of colors, so that adjacent vertices have different colors. In other words, the \emph{monochromatic connected components} (or \emph{monochromatic components}, for short) must form an independent set, i.e., they must be single vertices. A notable result on this topic is the so-called \emph{$4$-color theorem}~\cite{DBLP:conf/ascm/Gonthier07}, which states that four colors suffice to color the vertices of any planar graph such that adjacent vertices are of different colors; deciding whether a planar graph can be colored with three colors such that adjacent vertices are of different colors is a well-known NP-complete problem~\cite{DBLP:books/fm/GareyJ79}. 

\smallskip

In this paper, we continue the study of a variant of the vertex coloring problem, where each monochromatic component may be formed by more than one vertex but we require its cardinality to be ``small''~\cite{DBLP:conf/focs/KleinbergMRV97}. Formally, for a graph $G$ and an integer $t$ we denote by $mcc_t(G)$ the smallest integer $m$ such that there exists a coloring of the vertices of $G$ with $t$ colors so that any monochromatic component of $G$ has at most $m$ vertices. The classical vertex coloring is a restriction of this problem, in which one seeks for the smallest $t$ so that $mcc_t(G)=1$.

This problem is closely related to the well-known HEX lemma~\cite{Gale79}, which in its simplest form states that in any $2$-coloring of a $\sqrt{n} \times \sqrt{n}$ triangulated grid there exists either a monochromatic path from the left to the right side of the grid in one of the colors or a monochromatic path from the top side to the bottom side in the other color; see also~\cite{DBLP:journals/siamdm/MatousekP08}. As a consequence, not all (planar) graphs of bounded degree admit $2$-colorings in which all monochromatic components are of bounded size. Lov\'asz~\cite{Lovasz1975269} proved that any (not necessarily planar) cubic graph $G$ admits a $2$-coloring in which all monochromatic components are of size at most two, i.e., $mcc_2(G)\leq 2$. Also, such graphs admit $2$-colorings in which one color-class is an independent set, while the other one induces monochromatic components of size at most 750~\cite{DBLP:journals/jct/BerkeS07}.

Alon et al.~\cite{DBLP:journals/jct/AlonDV03} generalized these results to graphs with maximum degree 4 by proving that such graphs admit $2$-colorings in which all monochromatic components are of size at most $57$. Haxell et al.~\cite{DBLP:journals/jct/HaxellST03} improved this bound to $6$ and also proved that every graph with maximum degree 5 can be $2$-colored so that all monochromatic components have size at most 20000, a bound which was later reduced to 1908 by Berke in his Ph.D.\ thesis~\cite{BerkeThesis}. Linial et al.~\cite{DBLP:journals/endm/LinialMST07} showed that there exist $7$-regular $n$-vertex graphs whose monochromatic components are of size $\Omega(n)$ in any $2$-coloring. Note that for $6$-regular graphs the corresponding order of magnitude ranges between $\sqrt{n}$ and $n$~\cite{DBLP:journals/jct/AlonDV03,Gale79}. Also, Linial et al.~\cite{DBLP:journals/endm/LinialMST07} proved that $mcc_2(G)=O(n^{2/3})$, where $G$ is planar or belongs to a minor-closed class.

As for three or more colors, Kleinberg et al.~\cite{DBLP:conf/focs/KleinbergMRV97} and Alon at al.~\cite{DBLP:journals/jct/AlonDV03} constructed planar graphs that do not admit $3$-colorings so that each monochromatic component has bounded size. Esperet et al.~\cite{DBLP:journals/cpc/EsperetJ14} proved that there is a function $f: \mathbb{N} \rightarrow \mathbb{N}$ so that every planar graph with maximum degree $\Delta$ admits a $3$-coloring in which each monochromatic component has size at most $f(\Delta)$, where $f(\Delta) = (14\Delta)^{23\Delta+8}$. When a fixed number $t$ of colors is available and $G$ is of bounded tree-width, Linial et al.~\cite{DBLP:journals/endm/LinialMST07} proved that $mcc_t(G)=O(n^{1/t})$.

Note that there is a fairly rich literature on colorings of graphs where each monochromatic component is acyclic (see, e.g.,~\cite{DBLP:journals/dm/DingO96}) or has a small diameter (see, e.g.,~\cite{DBLP:journals/combinatorica/LinialS93}), or has small degree (see, e.g.,~\cite{DBLP:journals/jgt/CowenGJ97}).

\medskip\noindent\textbf{Our contribution.} In this paper, we study $2$- and $3$-colorings for two important subclasses of planar graphs:
\begin{inparaenum}[(i)]
\item Maximal outerplanar graphs, and
\item Complete planar 3-trees.
\end{inparaenum}
Note that both these classes are meaningful subclasses of planar graphs which are of interest in computational complexity theory, since many NP-complete graph problems are solvable in linear time on these graphs (see, e.g.,~\cite{Takamizawa:1982:LTC}). We also note that both of these classes of graphs are of bounded tree-width, which implies that any graph $G$ belonging to them has $mcc_2(G) = O(n^{1/2})$ and $mcc_3(G) = O(n^{1/3})$~\cite{DBLP:journals/endm/LinialMST07}. We present improvements upon these bounds: 
\begin{itemize} 
\item For an outerplanar graph $G$ with maximum degree $\Delta$, Berke has shown that $mcc_2(G) \leq 2 (\Delta-1)$~\cite{BerkeThesis}. Observe that this upper bound drastically improves the general one by Esperet et al.~\cite{DBLP:journals/cpc/EsperetJ14} for the class of outerplanar graphs. We present an efficient dynamic-programming based algorithm which for a given outerplanar graph~$G$ determines the exact value of $mcc_2(G)$. We further show improved upper bounds for special subclasses of outerplanar graphs. We note that outerplanar graphs admit proper $3$-colorings; see, e.g.,~\cite{Proskurowski:1986:EVE}.

\item For a complete planar 3-tree $G$, we prove that $mcc_2(G) = O(n^{0.387})$ and $mcc_3(G) = O(\log n)$. For two colors we also give a corresponding lower bound of $\Omega(n^{0.315})$. However, our bounds do not transfer to general planar $3$-trees; Linial~\cite{DBLP:journals/endm/LinialMST07} presented infinitely many planar graphs, that are subgraphs of planar $3$-trees and meet the aforementioned upper bounds of $O(n^{1/2})$ and $O(n^{1/3})$ for $2$- and $3$-colorings, respectively. 
\end{itemize}

The rest of the paper is organized as follows. Preliminary definitions are given in Section~\ref{sec:preliminaries}. The results about maximal outerplanar graphs are presented in Section~\ref{sec:outerplanar}, while those about complete planar $3$-trees in Section~\ref{sec:planar-3-trees}. Conclusions and open problems can be found in Section~\ref{sec:conclusions}.

%=================================================================
\section{Preliminaries}
\label{sec:preliminaries}
%================================================================= 

Let $G=(V,E)$ be a graph. $G$ is \emph{$1$-connected} if there is a path between any two vertices of $G$. $G$ is \emph{$k$-connected}, for $k \geq 2$, if the removal of $k-1$ vertices leaves $G$ $1$-connected. A $2$-connected graph is also called \emph{biconnected}.

A \emph{planar drawing} $\Gamma$ of $G$ is a drawing in the plane such that each vertex $v\in V$ is drawn as a distinct point $p_v$; each edge $(u,v)\in E$ is drawn as a simple curve connecting $p_u$ and $p_v$ and no two edges intersect in $\Gamma$ except at their common endpoints. $\Gamma$ divides the plane into topologically connected regions called \emph{faces}; the unbounded face is called the \emph{outer face}, all other faces are called \emph{internal faces}. Each face is described by the circular sequence of vertices encountered while moving clockwise along its boundary. The description of the set of faces determined by a planar drawing of $G$ is called a \emph{planar embedding} of $G$. Graph $G$ together with a given planar embedding is an \emph{embedded planar graph}.

An \emph{outerplanar graph} is a planar graph that admits a planar embedding such that all vertices are on the outer face. A \emph{maximal outerplanar graph} is an outerplanar graph to which no edge can be added while preserving outerplanarity (see Figure~\ref{fi:maximal outerplanar}).

\begin{figure}[ht]
\begin{subfigure}{0.45\columnwidth}
\centering
\includegraphics[width=.5\linewidth,page=1]{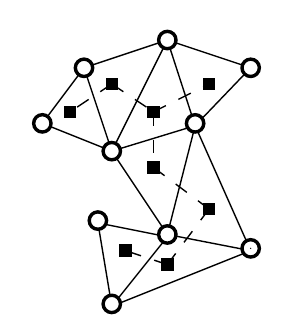}
\caption{}
\label{fi:maximal outerplanar}
\end{subfigure}
\begin{subfigure}{0.45\columnwidth}
\centering
\includegraphics[width=.5\linewidth,page=2]{figures/preliminaries}
\caption{}
\label{fi:complete_planar_3tree}
\end{subfigure}
\caption{
(a)~A maximal outerplanar graph $G$ and its weak dual represented by squared vertices and dashed edges. 
(b)~A complete planar $3$-tree with $2$ levels.}
\label{fi:coloring_one_fan}
\end{figure}

The \emph{dual} of an embedded planar graph $G$ is a graph that has a vertex for each face of $G$ and an edge between two face-vertices if and only if the corresponding faces share an edge in the embedding of $G$. The \emph{weak dual} of $G$ is the subgraph of the dual of $G$ obtained by omitting the face-vertex of the outer face of $G$ (see Figure~\ref{fi:maximal outerplanar}). The weak dual of an outerplanar graph is a forest in general. For a biconnected outerplanar graph the weak dual is a tree, and vice versa.

A complete planar $3$-tree is a graph formally defined as follows. A $3$-cycle is a planar $3$-tree with $0$ levels. A complete planar $3$-tree with $k\geq 1$ levels is obtained from one with $k-1$ levels by inserting a vertex in every internal face and connecting it to all of its vertices (see Figure~\ref{fi:complete_planar_3tree}).

%=================================================================
\section{Outerplanar graphs}
\label{sec:outerplanar}
%=================================================================
In this section, we focus on outerplanar graphs. We start with a dynamic programming based algorithm which for a given maximal outerplanar graph~$G$ determines the exact value of $mcc_2(G)$.

\begin{theorem}\label{thm:outerpathRecognition}
There exists a dynamic programming based algorithm which in polynomial time computes a $2$-coloring of a maximal outerplanar graph, such that the size of the largest monochromatic component is minimized.
\end{theorem}
\begin{proof}
Let $G$ be a maximal outerplanar graph on $n \geq 4$ vertices and assume without loss of generality that $G$ is embedded according to its outerplanar embedding. Since $G$ is maximal, $G$ is also biconnected, and it follows that the weak dual $\cal T$ of $G$ is a tree. We further assume that $\cal T$ is rooted at a leaf $\rho$ of it; note that $\cal T$ has at least two leaves, since $n \geq 4$. Our algorithm is mainly based on a definition of \emph{equivalent colorings}, which we use in a bottom-up traversal of $\mathcal{T}$ to maintain representative colorings from each class of equivalent colorings.

For a node $\mu$ of $\mathcal{T}$, we denote by $G(\mu)$ the subgraph of $G$ whose weak dual is the subtree of $\mathcal{T}$ rooted at $\mu$. Consider a node $\mu$ of $\mathcal{T}$ that is not the root of $\mathcal{T}$, that is, $\mu \neq \rho$. Let $\nu$ be the parent of $\mu$ in $\mathcal{T}$. Let also $f(\mu)$ and $f(\nu)$ be the faces of $G$ corresponding to nodes $\mu$ and $\nu$ of $\mathcal{T}$. Finally, let $(u,v)$ be the edge of $G$ shared by $f(\nu)$ and $f(\mu)$. We say that $(u,v)$ is the \emph{attachment edge} of $G(\mu)$ to $G(\nu)$, while its endpoints $u$ and $v$ are the \emph{poles} of $G(\mu)$ and $G(\nu)$. The attachment edge of root $\rho$ is any edge of face $f(\rho)$ that is incident to the outer face (since $G$ is biconnected and $\rho$ is a leaf, this edge always exists).

Consider a coloring $c$ of $G(\mu)$ with two colors, black and white, and let $c(B_\mu)$ and $c(W_\mu)$ be the sizes of the largest monochromatic components in black and white, respectively, that contain neither $u$ nor $v$. Let $c(u)$ and $c(v)$ be the colors of $u$ and $v$ in $c$, respectively. Let also $s(c(u))$ and $s(c(v))$ be the size of the monochromatic components containing $u$ and $v$ in $c$, respectively. Since $u$ and $v$ are adjacent, it follows that if $c(u)=c(v)$ holds, then $s(c(u))=s(c(v))$ also holds. We say that two colorings $c$ and $c'$ of $G(\mu)$ are \emph{equivalent} with respect to the attachment edge $(u,v)$ if and only if
\begin{inparaenum}[(i)]
\item $u$ has the same color in $c$ and $c'$, that is, $c(u)=c'(u)$,
\item $v$ has the same color in $c$ and $c'$, that is, $c(v)=c'(v)$,
\item the size of the monochromatic component containing $u$ is the same in $c$ and $c'$, that is, $s(c(u))=s(c'(u))$,
\item the size of the monochromatic component containing $v$ is the same in $c$ and $c'$, that is, $s(c(v))=s(c'(v))$,
\item the size of the largest black monochromatic component that contains neither $u$ nor $v$ is the same in $c$ and $c'$, that is, $c(B_\mu)=c'(B_\mu)$, and
\item the size of the largest white monochromatic component that contains neither $u$ nor $v$ is the same in $c$ and $c'$, that is, $c(W_\mu)=c'(W_\mu)$.
\end{inparaenum}
This definition of equivalence determines a partition of the colorings of $G(\mu)$ into a set of equivalence classes.

To store a representative from each class of equivalent colorings of $G(\mu)$ we employ a table $T_\mu$, in which entry 
\[T_\mu[\kappa_1,\kappa_2,\lambda_1,\lambda_2,\sigma_1,\sigma_2]\] 
is true if and only if there exists a coloring of $G(\mu)$ in which
\begin{inparaenum}[(i)]
\item the color of $u$ is $\kappa_1$,
\item the color of $v$ is $\kappa_2$,
\item the size of the monochromatic component containing $u$ is $\lambda_1$,
\item the size of the monochromatic component containing $v$ is $\lambda_2$,
\item the size of the largest black monochromatic component that contains neither $u$ nor $v$ is $\sigma_1$, and
\item the size of the largest white monochromatic component that contains neither $u$ nor $v$ is $\sigma_2$.
\end{inparaenum}
Table $T_\mu$ is of size $O(\Delta^4)$, since $\kappa_1,\kappa_2 \in \{B,W\}$ and due to Berke~\cite{BerkeThesis} $0 \leq \lambda_1,\lambda_2,\sigma_1,\sigma_2 \leq 2\Delta$ holds.

When visiting a node $\mu$ of $\mathcal{T}$ in the traversal of $\mathcal{T}$, we assume that for each of the two children $\mu_1$ and $\mu_2$ of $\mu$ in $\mathcal{T}$, we have computed a representative coloring from each class of equivalent colorings and that we have stored this information in $T_{\mu_1}$ and $T_{\mu_2}$. To compute $T_\mu$, we consider each pair of possible equivalence classes of colorings of $G(\mu_1)$ and $G(\mu_2)$ and check whether their combination yields a representative in some equivalence class of colorings of $G(\mu)$; initially, all entries of $T_\mu$ are false. For each of $T_{\mu_1}$ and $T_{\mu_2}$, consider the following entries:
\[T_{\mu_1}[\overline{\kappa}_1,\overline{\kappa}_2,\overline{\lambda}_1,\overline{\lambda}_2,\overline{\sigma}_1,\overline{\sigma}_2] \mbox{~and~} T_{\mu_2}[\hat{\kappa}_1,\hat{\kappa}_2,\hat{\lambda}_1,\hat{\lambda}_2,\hat{\sigma}_1,\hat{\sigma}_2].\] 
If $\overline{\kappa}_2 \neq \hat{\kappa}_1$, then we can ignore this pair and proceed to a next one, as the vertex \emph{shared} by $G(\mu_1)$ and $G(\mu_2)$ has different colors in the colorings of $G(\mu_1)$ and $G(\mu_2)$, and hence cannot yield a valid coloring for $G(\mu)$. So, we can assume without loss of generality that $\overline{\kappa}_2 = \hat{\kappa}_1$. To simplify the presentation, we assume that the color of the shared vertex of $G(\mu_1)$ and $G(\mu_2)$ is black. The case where this vertex is colored white is symmetric.

We proceed by considering two cases:
\begin{inparaenum}[(i)]
\item \label{c:neq} $\overline{\kappa}_1 \neq \hat{\kappa}_2$, and
\item \label{c:eq}  $\overline{\kappa}_1 = \hat{\kappa}_2$.
\end{inparaenum}
In Case~(\ref{c:neq}), we first assume that $\overline{\kappa}_1$ and $\hat{\kappa}_2$ correspond to white and black colors, respectively. In this case, it holds that $\hat{\lambda}_1=\hat{\lambda}_2$. Since we have assumed that the color of the shared vertex of $G(\mu_1)$ and $G(\mu_2)$ is black, it follows that $T_\mu[\overline{\kappa}_1, \hat{\kappa}_2, \overline{\lambda}_1, \hat{\lambda}_1, \max \{\overline{\sigma}_1,\hat{\sigma}_1\},\max \{\overline{\sigma}_2,\hat{\sigma}_2\}]$ is true. We cope with the case where $\overline{\kappa}_1$ and $\hat{\kappa}_2$ correspond to black and white colors, respectively, symmetrically. This completes our analysis for Case~(\ref{c:neq}).

In Case~(\ref{c:eq}), we again consider two subcases, namely, either both $\overline{\kappa}_1$ and $\hat{\kappa}_2$ correspond to black or both to white. Note that these cases are not symmetric, as we have assumed that the color of the shared vertex of $G(\mu_1)$ and $G(\mu_2)$ is black. In the former case, it holds that $\overline{\lambda}_1=\overline{\lambda}_2$ and $\hat{\lambda}_1=\hat{\lambda}_2$. Since we have assumed that the color of the shared vertex of $G(\mu_1)$ and $G(\mu_2)$ is black, it follows that $T_\mu[\overline{\kappa}_1, \hat{\kappa}_2, \overline{\lambda}_1+\hat{\lambda}_1-1, \overline{\lambda}_1+\hat{\lambda}_1-1, \max \{\overline{\sigma}_1,\hat{\sigma}_1\},\max \{\overline{\sigma}_2,\hat{\sigma}_2\}]$ is true. To complete the description of our case analysis, it remains to consider the more involved case where both $\overline{\kappa}_1$ and $\hat{\kappa}_2$ correspond to white. Since the color of the shared vertex of $G(\mu_1)$ and $G(\mu_2)$ is black, it follows that it may form a new monochromatic component (in black) that is not incident to the poles of $G(\mu)$. Hence, we set $T_\mu[\overline{\kappa}_1, \hat{\kappa}_2, \overline{\lambda}_1, \hat{\lambda}_2, \max \{\overline{\sigma}_1,\hat{\sigma}_1,\overline{\lambda}_2+\hat{\lambda}_1\}-1,\max \{\overline{\sigma}_2,\hat{\sigma}_2\}]$ to true. This completes our analysis for Case~(\ref{c:eq}).

At the end of the bottom-up traversal of $\mathcal{T}$, we search in  $T_\rho$ for the entry $T_\rho[\kappa_1,\kappa_2,\lambda_1,\lambda_2,\sigma_1,\sigma_2]$ which minimizes the maximum of $\lambda_1$, $\lambda_2$, $\sigma_1$ and $\sigma_2$. The corresponding coloring of $G$ can be easily constructed by traversing $\mathcal{T}$ top-down, and by following the choices performed during the bottom-up visit. Note that for a node $\mu$ of $\mathcal{T}$ the computation of table $T_\mu$ can be done in $O(\Delta^8)$, since each of the tables $T_{\mu_1}$ and $T_{\mu_2}$ of the children $\mu_1$ and $\mu_2$ of $\mu$ are of size $O(\Delta^4)$. Hence, the time-complexity of our algorithm is $O(n \cdot \Delta^8)$, which is linear when $G$ has bounded degree.
\end{proof}

In the rest of this section, we present two side-results that form improvements on the general upper bound of $2 (\Delta-1)$ by Berke~\cite{BerkeThesis} for special subclasses of outerplanar graphs with maximum degree $\Delta$. 

The first subclass of outerplanar graphs that we consider is the class of snowflakes. Intuitively, a \emph{snowflake} is a biconnected maximal outerplanar graph whose weak dual is a complete binary tree. Formally, a snowflake $S_h$ of \textit{height} $h$ is defined recursively as follows (see Figure~\ref{fig:snowflake_coloring} for an example). For $h=0$, $S_h$ is a $3$-cycle, i.e., 3 vertices and 3 edges between them. For $S_h$ with $h > 0$, we take $S_{h-1}$ and extend it by attaching on each edge of the outer face a path of length 2. The new vertices are \emph{at height} $h$ and the length of the outer face is doubled.
When we add the set of vertices at height $h$, for $h>1$ they are adjacent to two vertices that we call \emph{ancestors}: One of them is at height $h-1$, and is denoted by $p(v)$, while the other has height $< h-1$. The height $h$ of the snowflake is related to the maximum degree of any vertex, namely, $deg(v) = 2h+2$ where $v$ is any vertex at height $0$.

\begin{figure}[ht]
\centering
\includegraphics[width=.25\linewidth,page=2]{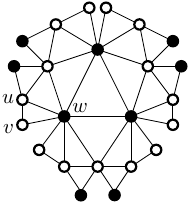}
\caption{A snowflake $S_h$ of height $h=3$. Vertex $v$ is at height $3$ and has $u$ and $w$ as its ancestors; $u=p(v)$ is at height $2$, while $w$ is at height $0$. The vertices of $S_3$ are colored according to Theorem~\ref{thm:snowflake}.}
\label{fig:snowflake_coloring}
\end{figure}

\begin{theorem}\label{thm:snowflake}
For a snowflake $S$ of maximum degree $\Delta \geq 4$, $mcc_2(S) \leq \Delta-3$ holds.
\end{theorem}
\begin{proof}
Assume that $S=S_h$ has height $h$. We color its vertices as follows: Initially, the vertices of height $0$ are colored with the same color. For a vertex $v$ of height $k$, we consider the colors of its two ancestors: If they both have the same color, we give to vertex $v$ the opposite color, otherwise we assign to $v$ the color of vertex $p(v)$.

We claim that each monochromatic component is a path. By construction, the edges between monochromatic vertices are of the form $(v,p(v))$, which we consider to be directed towards $p(v)$. Hence, they form monochromatic trees with direction in decreasing height. Note that if a vertex $v$ has two incoming edges, then its two ancestors both have different color than $v$. So, we can conclude that the monochromatic components consist of at most two height-decreasing paths that end in a common vertex. Since the paths have length at most  $h-1$, the component has size $2(h-1)+1$, which is equal to $\Delta-3$. That concludes the proof.
\end{proof}

The second subclass of outerplanar graphs that we consider is the class of outerpaths. Formally, an \emph{outerpath} is a biconnected maximal outerplanar graph whose weak dual is a path. Before we proceed with the description of our coloring scheme, we introduce some necessary definitions and notation. Let $\mathcal{O}$ be an outerpath (see Figure~\ref{fi:outerpath_coloring} for an example). We call \emph{spine vertices} of $\mathcal{O}$ the vertices $v_1, v_2, \ldots , v_k$ that have degree at least four in $\mathcal{O}$ and, for the sake of simplicity, we denote by $v_0$ and $v_{k+1}$ the two (unique) vertices of degree two of $\mathcal{O}$. Consider a spine vertex $v_i$ of $\mathcal{O}$ ($1\leq i\leq k$) and denote by $e'$ and $e''$ the edges $(v_i,v_{i-1})$ and $(v_i,v_{i+1})$ of $\mathcal{O}$, respectively. We call \emph{fan} $f_i$ the subgraph of $\mathcal{O}$ induced by $v_i$, $v_{i-1}$, $v_{i+1}$ and by all the vertices $w_j$ such that edge $(v_i,w_j)$ is between $e'$ and $e''$ in counterclockwise (clockwise, resp.) order around $v_i$ if $i$ is odd (even, resp.). Denote by $|f_i|$ the number of vertices of $f_i$. Observe that $|f_i|\geq 3$ for each $2\leq i\leq k-1$, $|f_i|\geq 4$ for $i=\{1,k\}$, and that $f_i$ and $f_{i+1}$ share the two spine vertices $v_i$ and $v_{i+1}$ for each $1\leq i \leq k-1$. We say that a vertex $v$ is \emph{isolated} if all its adjacent vertices have different color.

The last tool that we need for our analysis is the following lemma, proven in several papers; see, e.g.,~\cite{BerkeThesis}. Recall that a \emph{wheel} $W$ on $n$ vertices is a graph formed by $n-1$ vertices connected in a cycle, called the \emph{rim} of $W$, and a so-called \emph{central vertex} that has an edge, called \emph{spoke}, towards each vertex of the rim.
\begin{lemma}[Berke~\cite{BerkeThesis}]\label{lem:wheel_graph}
For a planar wheel W on $n$ vertices, $mcc_2(W) = \Theta(\sqrt{n})$ holds.
\end{lemma}

\noindent We are now ready to prove the last result of this section.

\begin{figure}[ht]
\centering
\includegraphics[scale=0.8,page=1]{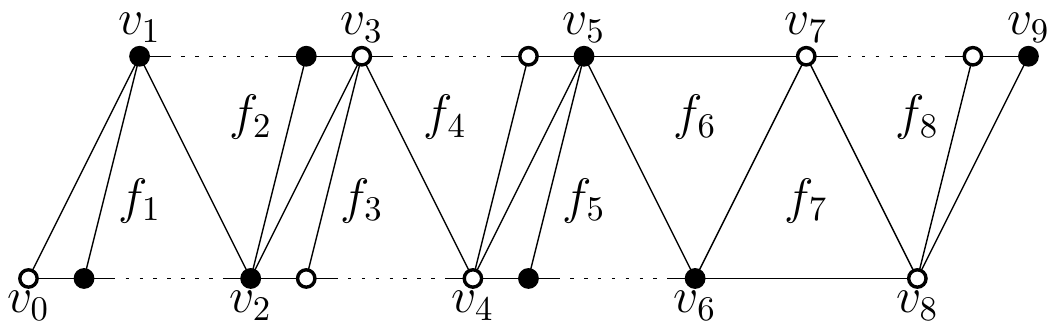}
\caption{An outerpath $\mathcal{O}$ with spine vertices $v_1 \dots v_8$. The vertices of $\mathcal{O}$ are colored according to Theorem~\ref{thm:outerpath}.}
\label{fi:outerpath_coloring}
\end{figure}

\begin{theorem}\label{thm:outerpath}
For an outerpath $\mathcal{O}$ of maximum degree $\Delta$, $mcc_2(\mathcal{O})=\Theta (\sqrt{\Delta})$ holds.
\end{theorem}
\begin{proof}
We color $\mathcal{O}$ with a technique similar to the one used in Lemma~\ref{lem:wheel_graph} for a wheel. Note that the coloring of a wheel with $n$ vertices is obtained by coloring its center and $\lfloor \sqrt{n} \rfloor$ of its neighbors black so that the remaining vertices, which are colored white, are partitioned into blocks of size $\sqrt{n}$ or $\sqrt{n}+1$. For the case of an outerpath, we first color the spine vertices and then we use the idea of the wheel to color its fan vertices. More precisely, we do the following (see Figure~\ref{fi:outerpath_coloring}). Without loss of generality, set the color of $v_1$ black.

\noindent \textbf{Spine vertices:} Starting from $v_1$, consider pairs of consecutive spine vertices; color the vertices of each pair with the same color and alternate the color of two consecutive pairs. So, $v_1$ and $v_2$ are colored black, $v_3$ and $v_4$ are colored white, $v_5$ and $v_6$ are colored black, and so on. For the two vertices $v_0$ and $v_{k+1}$ of degree two do the following: Color $v_0$ white; if $v_k$ and $v_{k-1}$ have the same color, then color $v_{k+1}$ with the other color, otherwise color $v_{k+1}$ with the same color as $v_k$.

\noindent \textbf{Fan vertices:} If $|f_i|=3$ there is no vertex to color. So, assume $|f_i|\geq 4$ and assume without loss of generality that $v_i$ is colored black. Then, either $v_{i-1}$ or $v_{i+1}$ is white. Suppose that $v_{i-1}$ is white. Color black the unique vertex adjacent to both $v_i$ and $v_{i-1}$. The remaining vertices of each fan are colored using the technique of Lemma~\ref{lem:wheel_graph}.

\medskip We now prove that the subgraph $\mathcal{O}_\ell$ induced by the fans $f_1$, $f_2$, $\dots$, $f_\ell$ has $mcc_2(\mathcal{O}_\ell)=O(\sqrt{\Delta})$ with the following additional property: If $\ell$ is odd, then $\mathcal{O}_\ell$ has vertex $v_0$ isolated, while if $\ell$ is even, $\mathcal{O}_\ell$ has both vertices $v_0$ and $v_{\ell+1}$ isolated. The proof is by induction on the number of spine vertices of $\mathcal{O}$. 

If $\ell=1$, $\mathcal{O}_1$ is composed of fan $f_1$ and $|f_1|\geq 4$. By construction, $v_1$ is black, $v_0$ is white, the other vertex adjacent to $v_0$ (call it $w_1$) is black, and $v_2$ is black (see Figure~\ref{fig:base_cases_coloring_1}). The other vertices of $f_1$ are colored according to the technique of Lemma~\ref{lem:wheel_graph}. Observe that $v_0$ is isolated. Consider the subgraph $\mathcal{O}_1'=\mathcal{O}_1\setminus \{v_0\}$. Since $v_0$ is isolated, $mcc_2(\mathcal{O}_1)= mcc_2(\mathcal{O}_1')$. Since $\mathcal{O}_1$ has $\Delta +1$ vertices, $\mathcal{O}_1'$ has $\Delta$ vertices. Now consider the graph $W$ obtained from $\mathcal{O}_1'$ by making vertices $w_1$ and $v_2$ coincident. $W$ is a wheel with $\Delta-1$ vertices (see Figure~\ref{fig:base_cases_coloring_2}). Then $mcc_2(W)=O(\sqrt{\Delta})$. Since $mcc_2(\mathcal{O}_1) = mcc_2(\mathcal{O}_1') \leq mcc_2(W)+1$ (in $\mathcal{O}_1'$ vertices $w_1$ and $v_2$ are both black and distinct), it follows that $mcc_2(\mathcal{O}_1)=O(\sqrt{\Delta})$. 

\begin{figure}[h]
\centering
\begin{subfigure}{0.32\textwidth}
\centering
\includegraphics[scale=0.8,page=2]{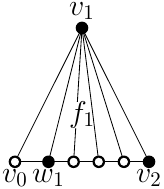}
\caption{}
\label{fig:base_cases_coloring_1}
\end{subfigure}
\begin{subfigure}{0.32\textwidth}
\centering
\includegraphics[scale=0.8,page=3]{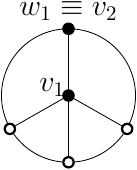}
\caption{}
\label{fig:base_cases_coloring_2}
\end{subfigure}
\begin{subfigure}{0.32\textwidth}
\centering
\includegraphics[scale=0.8,page=4]{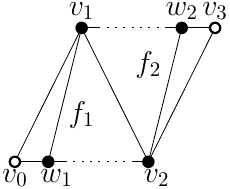}
\caption{}
\label{fig:base_cases_coloring_3}
\end{subfigure}
\caption{(a)-(b)~$\ell=1$; (c)~$\ell=2$.}
\label{fig:base_cases_coloring}
\end{figure}

If $\ell=2$, $\mathcal{O}_2$ is composed of fans $f_1$ and $f_2$ (see Figure~\ref{fig:base_cases_coloring_3}). By construction, $v_1$ and $v_2$ are black, while $v_0$ and $v_3$ are white. Observe that $|f_1| \geq 4$, while $|f_2| \geq 3$. If $f_2$ has size three, $\mathcal{O}_2$ is composed of fan $f_1$ plus vertex $v_3$ that is isolated. Then $mcc_2(\mathcal{O}_2)=O(\sqrt{\Delta})$ by the same argument as above. Otherwise, denote by $w_1$ and $w_2$ the two vertices adjacent to both $v_1$ and $v_0$ and to both $v_2$ and $v_3$, respectively. Vertices $w_1$ and $w_2$ are colored black. Observe that $v_0$ and $v_3$ are isolated and that $v_1$ and $v_2$ are shared by $f_2$ and $f_1$. Denote by $\mathcal{O}_2'$ ($\mathcal{O}_2''$) the subgraph of $\mathcal{O}_2$ induced by $f_1 \setminus \{v_0\}$ ($f_2 \setminus \{v_3\}$). We have that $mcc_2(\mathcal{O}_2'\cup \mathcal{O}_2'')\leq mcc_2(\mathcal{O}_2')+mcc_2(\mathcal{O}_2'')$. Using the same argument as for the case $\ell=1$ for both  $\mathcal{O}_2'$ and  $\mathcal{O}_2''$ and considering the (worst) case where both $v_1$ and $v_2$ have degree $\Delta$, it follows that $mcc_2(\mathcal{O}_2')=mcc_2(\mathcal{O}_2'')$ which is $O(\sqrt{\Delta})$, and then $mcc_2(\mathcal{O}_2'\cup \mathcal{O}_2'')= O(\sqrt{\Delta})$.

For the inductive hypothesis, assume that the statement holds for $\mathcal{O}_\ell$, with $\ell>2$. We prove that it also holds for $\mathcal{O}_{\ell+1}$. If $\ell$ is even, decompose $\mathcal{O}_{\ell+1}$ in the two subgraphs $\mathcal{O}_\ell$ and $f_{\ell+1}$. The induction hypothesis holds for $\mathcal{O}_\ell$, then $\mathcal{O}_\ell$ has both $v_0$ and $v_{\ell+1}$ isolated. Suppose that $v_{\ell+1}$ is white. Since $v_{\ell+1}$ is isolated, $v_\ell$ and $v_{\ell+1}$ have different color. By construction the (unique) vertex of fan $f_{\ell+1}$ adjacent to both $v_\ell$ and $v_{\ell+1}$ has the same color of $v_{\ell+1}$. Observe that the vertices of $f_{\ell+1}$ adjacent to $v_\ell$ have color different from $v_\ell$, and the vertices of $f_\ell$ adjacent to $v_{\ell+1}$ have color different from $v_{\ell+1}$ (recall that they are black). Then $mcc_2(\mathcal{O}_{\ell+1})=max \{mcc_2(\mathcal{O}_\ell),mcc_2(f_{\ell+1})\}$ which is $O(\sqrt{\Delta})$. If $\ell$ is odd, decompose $\mathcal{O}_{\ell+1}$ in the subgraph $\mathcal{O}_{\ell-1}$ and in the subgraph induced by fans $f_\ell$ and $f_{\ell+1}$. Since $\mathcal{O}_{\ell-1}$ is composed of an even number of fans, by induction it has both $v_0$ and $v_\ell$ isolated. Then, as before, $mcc_2(\mathcal{O}_{\ell+1})=max \{mcc_2(\mathcal{O}_{\ell-1}),mcc_2(f_\ell \cup f_{\ell+1})\}=O(\sqrt{\Delta})$.
To complete the proof, we note by Lemma~\ref{lem:wheel_graph} $mcc_2(\mathcal{O})=\Omega(\sqrt{\Delta})$ follows, since each fan of $\mathcal{O}$ can be regarded as a wheel.
\end{proof}

%=================================================================
\section{Complete Planar 3-trees}
\label{sec:planar-3-trees}
%=================================================================

In this section, we present $2$- and $3$-colorings for complete planar $3$-trees. However, before we proceed with the description of our approach, we will first introduce some important properties of planar $3$-trees, which we use in Section~\ref{subsec:planar-3-trees-algos}.

Let $T_k$ be a complete planar $3$-tree with $k\geq 0$ levels. By definition, $T_k$ is fully triangulated. Let $v_1$, $v_2$ and $v_3$ be the vertices of $T_k$ on its outer face, called \emph{outer vertices} of $T_k$. With a slight abuse of notation, we write $T_k=\{v_1,v_2,v_3\}$. As mentioned above, $T_k$, with $k>0$, is constructed from $T_{k-1}$ as follows: In every bounded face $f$ of $T_{k-1}$ we add a vertex $v$ in its interior and connect it to the three vertices of $f$, splitting $f$ into three ``smaller'' triangular faces. We say that $v$ is at \emph{level} $k$ and $v$ is the \emph{central vertex} of $f$. A vertex of maximum level is called a \emph{leaf-vertex} and it is adjacent to three \emph{empty-triangles}, that is, faces of $T_k$. By definition each triangle of $T_k$ contains in its interior another complete planar $3$-tree with $k'<k$ levels; if $k'>0$ we refer to such a triangle as \emph{non-empty}. Note that empty-triangles induce degenerate complete planar $3$-trees of $k'=0$ levels.

Outer vertices $v_i$ and $v_j$ (with $1\leq i,j\leq3$)  have $k$ common neighbors that form a so-called \emph{central path}  $P_{i,j}=w_1,w_2,\dots,w_k$, where $w_1$ is the central vertex of $T_1$, namely, $P_{i,j}$ is the path from the central vertex $w_1$ to the unique leaf-vertex $w_k$ that is a common neighbor of $v_i$ and $v_j$. For $1 \leq \ell < k$, vertices $v_i$, $w_\ell$ and $w_{\ell+1}$ of $P_{i,j}$ form a triangle, called \emph{side-triangle}. Note that side-triangles contain in their interior complete planar $3$-trees with $k-\ell-1$ levels; side-triangles are in principle non-empty, except for the case where $\ell=k-1$. Vertices $v_i$, $v_j$ and $w_k$ form an empty-triangle, called \emph{leaf-triangle}.

\begin{property}\label{prp:order}
$T_k$ has $n=3+\sum_{i=1}^{k}{3^{i-1}}=(3^k+5)/2$ vertices and $3^k+1$ triangular faces.
\end{property}

\begin{property}\label{prp:neighbors}
The number of neighbors for outer vertex $v_j$ is equal to $2+\sum_{i=1}^{k}{2^{i-1}}=2^{k}+1$, and the number of internal neighbors of all three outer vertices equals $3(2^k-1)-2-3(k-1)=3\cdot2^{k}-3k-2$.
\end{property}

\noindent We are now ready to present the main results of this section. As a last tool, we further need the following result, which is a direct consequence of Lemma~\ref{lem:wheel_graph}.

\begin{corollary}
\label{lem:wheel}
A (planar) graph $G$ that has as a subgraph a wheel on $\nu$ vertices has $mcc_2(G) = \Omega(\sqrt{\nu})$. 
\end{corollary}

Note that Lemma~\ref{lem:wheel_graph} holds only for two colors. For three colors, it is not difficult to see that $mcc_3(W)=1$ or $2$. One can argue similarly for a \emph{double wheel}. A double wheel on $n$ vertices consists of a wheel on $n-1$ vertices, and a second central vertex that is also adjacent to every vertex of the rim. It is not hard to see that $mcc_3(D)=1$ or $2$ holds for a double wheel $D$ as well. Note that in such a coloring the two central vertices always have the same color. However, for the case where the two centers must have different colors, one can see, following an approach similar to the one of the proof of Lemma~\ref{lem:wheel_graph}, that there exists a monochromatic component of size $\Omega(\sqrt{n})$.

%=================================================================
\subsection{Colorings of Complete Planar 3-trees}
\label{subsec:planar-3-trees-algos}
%=================================================================

We are now ready to present our main results for complete planar $3$-trees. In Theorem~\ref{thm:3trees_mcc3} we focus on $3$-colorings, while Theorem~\ref{thm:3trees_mcc2} regards $2$-colorings. 

\begin{theorem}\label{thm:3trees_mcc3}
For a complete planar $3$-tree $T$ on $n$ vertices, $mcc_3(T) = \Theta(\log n)$ holds.
\end{theorem}
\begin{proof}
To prove the upper bound, assume that $T=T_k=\{v_1,v_2,v_3\}$ has $k$ levels. We describe a recursive coloring algorithm (see Figure~\ref{fi:completePlanar3Tree} for an example), which maintains the following simple invariant: The outer vertices of the graph to be colored are bicolored. To this end, we initially color vertices $v_1$ and $v_2$ with the color $c_1$, and vertex $v_3$ with the color $c_2$. Clearly, this conforms with our invariant. Then we color the central vertex $v$ and the three central paths $P_{i,j}$ for $1\leq i<j\leq 3$, with $v$ as their common end-vertex, with $c_3$. The induced subgraph with vertices $v_1$, $v_2$, $v_3$ and the vertices of the three central paths contains a monochromatic component of size $3k-2$, and splits $T$ into inner triangles. Two of the three leaf-triangles have all three colors on their vertices but they contain no other vertex in their interior. Side-triangles have their vertices colored with two colors (since two vertices belong to a central path and are colored with $c_3$). The algorithm recursively computes the colors of the vertices in the interior of the non-empty side-triangles. Note that every side-triangle is a complete planar $3$-tree with $\ell<k$ levels and its outer vertices have two colors, which conforms with our invariant. Hence, at each recursive step of our algorithm, a monochromatic component of size $3\ell-2$ is created. By construction, each monochromatic component inside the side-triangle is either not adjacent to the three outer vertices, or it has different color. This property ensures that no two monochromatic components of the same color, created at two different steps, are adjacent. Hence, the size of the largest monochromatic component is $3k-2=O(\log n)$. In order to prove the lower bound, we need the following claim.

\begin{figure}[tb]
\centering
\includegraphics[width=.42\columnwidth,page=1]{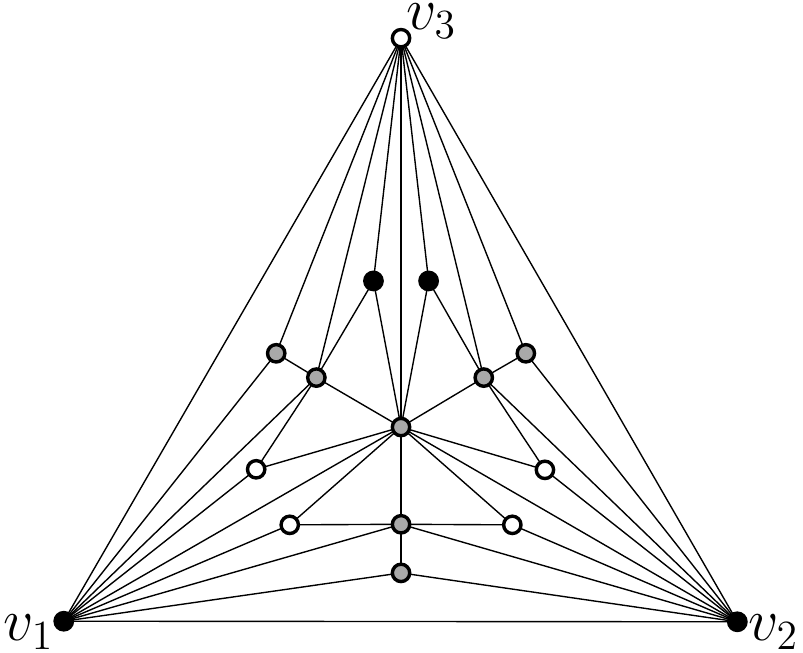}
\caption{A complete planar $3$-tree with $3$ levels. Vertices are colored according to Theorem~\ref{thm:3trees_mcc3}.}
\label{fi:completePlanar3Tree}
\end{figure}

\begin{claim}
Let $T=\{v_1,v_2,v_3\}$ be a complete $3$-colored planar $3$-tree with $k$ levels, where $v_1$, $v_2$ and $v_3$ have all different colors. Then, there exists a monochromatic path of length at least $k/3$, with one endpoint on the outer face of $T$.
\end{claim}
\begin{proof}
For $k=0$ the claim trivially holds, so we assume that $k>0$. Denote by $P_i$ a degenerate path consisting only of vertex $v_i$, for $i=1,2,3$. We shall prove that at least one of these three paths can be extended to a monochromatic path of length at least $k/3$. 
Let $w$ be the central vertex of $T$. Since $v_1$, $v_2$ and $v_3$ have different colors, $w$ must have the same color with exactly one of the three outer vertices. Assume without loss of generality that $w$ has the same color as $v_1$. This implies that we can extend monochromatic path $P_1$ by adding vertex $w$, that is $P_1=v_1,w$. Also, vertex $w$ forms a $3$-colored triangle $T_1$ together with outer vertices $v_2$ and $v_3$ of $T$. Now $T_1=\{v_2,v_3,w\}$ is a complete planar $3$-tree with $k-1$ levels and vertices of its outer face have different colors. Furthermore, paths $P_1$, $P_2$ and $P_3$ contain four vertices in total.
Following a similar approach, from the central vertex of $T_1$, one can have a complete planar $3$-tree $T_2$ with $k-2$ levels, whose outer vertices have different colors. Furthermore, the central vertex of $T_1$ extends one of the three monochromatic paths by one vertex. Repeating the above procedure $k$ times, we can create a chain of nested complete planar $3$-trees $T$, $T_1$, \dots, $T_k$, where their outer vertices have different colors. At each step, the central vertex is added to one of the three monochromatic paths. Paths $P_1$, $P_2$ and $P_3$ have a total of $3+k$ vertices. Hence, at least one of them has length at least $k/3$, and the claim holds.
\end{proof}

From the above claim, in order to prove the lower bound, it suffices to find a subgraph $T'$  of $T$ that is a complete planar $3$-tree with $\Omega(\log n)$ levels, whose outer vertices have all three colors. Consider outer vertex $v_1$ of $T$ and let $N_\ell(v_1)$ be its neighboring vertices at level at most $\ell$, where $\ell=2\log\log _3n$. Then the vertices $v_1$, $v_2$, $v_3$ and $N_\ell(v_1)$ form a wheel $W$ with the center vertex $v_1$. By Property~\ref{prp:neighbors}, $|N_\ell(v_1)|=2^\ell+1$, and therefore wheel $W$ has $n_W=1+(2^{\ell}+1)$ vertices. Let $c_1$ be the color of vertex $v_1$. We consider two cases depending on whether there exist two consecutive vertices along the rim of $W$ with colors $c_2$ and $c_3$. Suppose first that this is not the case, and vertices of $W$ with color $c_2$ are not adjacent to vertices with color $c_3$. This implies that if one identifies colors $c_2$ and $c_3$ to one color, the size of all monochromatic components in the derived $2$-coloring of $W$ remains the same. By Lemma~\ref{lem:wheel_graph}, we have that $W$ contains a monochromatic component of size $\Omega(\sqrt{n_W})$. For $\ell=2\log\log _3n$, it follows that $W$, and therefore $T$, contains a monochromatic component of size $\Omega(\sqrt{n_W})=\Omega(\log n)$. In the second case, there exist two consecutive vertices on the cycle of $W$ with colors $c_2$ and $c_3$. These two vertices, together with $v_1$, form a non-empty triangle $T'$ that is a complete planar $3$-tree with at least $k'\geq k-\ell$ levels. Since the three outer vertices of $T'$ are of different colors, the previous claim assures that there exists a monochromatic path of length at least $k'/3$. For $\ell=2\log\log _3n$, it follows that there exists a monochromatic path of length $\Omega(\log n)$. This completes our proof.
\end{proof}

\noindent We now adjust the technique used in the proof of Theorem~\ref{thm:3trees_mcc3} to obtain the following result.

\begin{theorem}\label{thm:3trees_mcc2}
For a complete planar $3$-tree $T$ on $n$ vertices, the following holds:

$$mcc_2(T) = O(n^{1/{\log 6}})=O(n^{0.387}).$$
\end{theorem}
\begin{proof}
The general idea is to color the first $\ell$ levels with color $c_1$, their neighbors with color $c_2$ and to use recursion. Assume that $T=T_k$ has $k$ levels. We color the three outer vertices plus the vertices of the first $\ell$ levels with color~$c_1$. This yields a monochromatic component of size $n_1=(3^\ell+5)/2$ (by Property~\ref{prp:order}). This coloring implies a ``subdivision'' of the remaining vertices into $3^\ell$ components, that are planar $3$-trees with $k-\ell$ levels. In each of the smaller $3$-trees,  say $T_{k-\ell}$, we use color $c_2$ for the neighbors of the outer vertices (outer vertices have already color $c_1$). By Property~\ref{prp:neighbors} we create components of size $n_2=3\cdot2^{k-\ell}-3(k-\ell)-2$. 

After this step, uncolored vertices form smaller components in the interior of complete planar $3$-trees that have fewer than $k-\ell$ levels and their outer vertices have color $c_2$. Note that there exist non-monochromatic triangles, however we claim that they are empty. Assume that there is an uncolored vertex $u$ inside a non-monochromatic triangle $T_k$. At least one of the vertices of $T_k$ has color $c_1$ and it is therefore an outer vertex of a tree $T_{k-\ell}$. Hence $u$ is adjacent to this outer vertex, and therefore has already color $c_2$, a contradiction.

The remaining components are complete planar $3$-trees with $k'<k-\ell$ levels each. The outer vertices of $T_{k'}$ have color $c_2$ and we proceed by coloring the neighbors of its outer vertices with color $c_1$. Thus we create one monochromatic component with fewer than $n_2$ vertices. As before, uncolored vertices form smaller complete planar $3$-trees with even fewer levels than $k'$. Furthermore the property of empty non-monochromatic triangles still holds, and we can recursively color the remaining uncolored components, by alternatively using colors $c_1$ and $c_2$. At every recursive step, we create monochromatic components with fewer than $n_2$ vertices. Hence, $mcc_2(T)\leq \max\left\{n_1, n_2\right\}$, which is minimized for $n_1=n_2$, that is, $\frac{(3^\ell+5)}{2}=3\cdot2^{k-\ell}-3(k-\ell)-2$. So, we have:

\[\ell=1+\frac{k}{\log6}+\frac{\log(1-2^{\ell-k-1}(3+2(k-\ell)))}{\log6}=\frac{k}{\log6}+C, \]

\noindent where $C=1+o(1)$. Then: 
$$mcc_2(T) \leq \frac{(3^\ell+5)}{2}=\frac{(3^C\cdot3^{\frac{k}{\log6}}+5)}{2}$$ 

\noindent By Property~\ref{prp:order}, it follows that $3^k=2n-5$ and thus $mcc_2(T)= O(n^{\frac{1}{\log6}})$.
\end{proof}

\begin{theorem}
For a complete planar $3$-tree $T$ on $n$ vertices, the following holds:

$$mcc_2(T) = \Omega(n^{\frac{1}{2} \log_3{2}})=\Omega(n^{0.315}).$$
\end{theorem}
\begin{proof}
The neighbors of each of the outer vertices form a wheel, each of which has size $2+1+2+\ldots+ 2^{\log_3{n}} = 2^{\log_3{n}+1} + 1 = \Omega(n^{\log_3{2}})$. Then, the lemma follows from Corollary~\ref{lem:wheel}.
\end{proof}

%=================================================================
\section{Conclusions and Open Problems}
\label{sec:conclusions}
%================================================================= 

In this paper, we studied $2$- and $3$-colorings and presented improved bounds on the size of the monochromatic components for internally triangulated outerplanar graphs and complete planar 3-trees. Several questions remain open: 

\begin{enumerate}[(i)] 
\item The study of other classes of graphs with bounded tree-width, which would give rise to improved bounds w.r.t.\ the ones of Linial et al.~\cite{DBLP:journals/endm/LinialMST07}, is of interest. 
\item Other classes of graphs that allow for an efficient computation of the largest monochromatic component are of importance. 
\item We believe that the class of $1$-planar graphs would be interesting in this context, as it is neither minor closed nor of bounded tree-width. No results are known for the size of the largest monochromatic component (for $2$- and $3$-colorings) of these graphs.
\end{enumerate}

\subsection*{Acknowledgments} We thank Tamara Mchedlidze for useful discussions.

\bibliography{references}
\bibliographystyle{abbrv}

\end{document}